%% file: socg.tex
\documentclass{sig-alternate-10}
\usepackage{booktabs}

\newtheorem{lemma}{Lemma}[section]
\newtheorem{theorem}{Theorem}[section]
\newtheorem{remark}{Remark}[section]

\begin{document}

\title{Improvements on the k-center problem for uncertain data}

\author{Sharareh Alipour   \and Amir Jafari}
\maketitle

\begin{abstract}

In real applications, there are situations where we need to model some problems based on uncertain data. This leads us to define an uncertain model for some classical geometric optimization problems and propose algorithms to solve them. 
In this paper, we study the $k$-center problem, for uncertain input.
In our setting, each uncertain point $P_i$  is located independently from other points in one of several possible locations $\{P_{i,1},\dots, P_{i,z_i}\}$ in a metric space with metric $d$, with specified probabilities and the goal is to compute $k$-centers $\{c_1,\dots, c_k\}$  that minimize the following expected cost
$$Ecost(c_1,\dots, c_k)=\sum_{R\in \Omega} prob(R)\max_{i=1,\dots, n}\min_{j=1,\dots k} d(\hat{P}_i,c_j)$$
here $\Omega$  is the probability space of all realizations 
$$R=\{\hat{P}_1,\dots, \hat{P}_n\}$$ of given uncertain points and 
$$prob(R)=\prod_{i=1}^n prob(\hat{P}_i).$$
 In restricted assigned version of this problem,
an assignment $A:\{P_1,\dots, P_n\}\rightarrow \{c_1,\dots, c_k\}$ is given for any choice of centers and the goal is to minimize 
$$Ecost_A(c_1,\dots, c_k)=\sum_{R\in \Omega} prob(R)\max_{i=1,\dots, n} d(\hat{P}_i,A(P_i)).$$ In unrestricted version, the assignment is not specified and the goal is to compute  $k$ centers $\{c_1,\dots, c_k\}$ and an assignment $A$ that minimize the above expected cost.

We give several improved constant approximation factor algorithms for the assigned versions of this problem in a Euclidean space and in a general metric space.
Our results significantly improve the results of \cite{guh} and generalize the results of \cite{wang} to any dimension.
Our approach is to replace a certain center point for each uncertain point and study the properties of these certain points.
The proposed algorithms are efficient and simple to implement.
\end{abstract}
\smallskip
\\
\noindent \textbf{Keywords.} $k$-center problem, uncertain points, approximation algorithm.

\input{samplebody-conf}




\bibliographystyle{abbrv} 
\bibliography{sigproc}

\end{document}

%% file: samplebody-conf.tex
\section{Introduction}
It is not surprising that in many real-world applications, we face uncertainty about the data.
Database systems should be able to handle and correctly process these uncertain data. Most of the time, we need to deal with optimization problems in data bases, such as data integration, streaming, cluster computing and sensor network applications that involve parameters and inputs whose values are known only with some uncertainty\cite{guh}. So, an important challenge for database systems is to deal with large amount of data with uncertainty.

In this paper we focus on a classical geometric optimization problem, $k$-center problem, for uncertain data. First, we introduce the uncertainty models and  the previous works, then we propose our algorithms for these models. 

\subsection*{Problem Statement}
In a metric space $X$ with metric $d$, the $k$-center problem for a set of (certain) points $\{P_1,\dots,P_n\}$ in $X$, asks for $k$ center points $C=\{c_1,\dots, c_k\}$ in $X$ that minimize the following cost
\begin{center}
$cost(c_1,\dots, c_k)=\max_{i=1,\dots,n}$  $d(P_i, C),$
\end{center} 
where $d(P_i,C)=\min_{c\in C} d(P_i,c)$.
When the points $P_1,\dots,P_n$ are uncertain, each point has a finite number of possible locations independently from the other points with given probabilities. More precisely, we are given a set $D=\{D_1,\dots,D_n\}$ of $n$ discrete and independent probability distributions. The $i$-th distribution, $D_i$ is defined over a set of $z_i$ possible locations $P_{i1},\dots,P_{i,z_i}\in X$. A probability $p_{ij}$ is associated with each location such that $\sum_j p_{ij}=1$ for every $i\in[n]=\{1,\dots, n\}$ and $j\in\{1,\dots, z_i\}$. Thus, the probabilistic points can be considered to be independent random variables $X_i$. The locations together with the probabilities specify their distributions $Pr[X_i=P_{ij}]=p_{ij}$ for every $i \in [n]$ and $j\in [z_i]$. A probabilistic set $Y$, consisting of the probabilistic points, is therefore a random variable. Let $z=\max\{z_1,\dots, z_n\}$ be the maximum number of possibilities for uncertain points.

For simplicity, we use the notation $\hat{P}_i$ for a realization of the uncertain point $P_i$ and the $prob(\hat{P}_i)$ for its probability. We let $\Omega$ denote the probability space of all realizations $R=\{P_{1j_1},\dots P_{nj_n}\}$ with $prob(R)=\prod_{i=1}^n prob(P_{i,j_i})$.
\\
There are three known versions of the $k$-center problem for uncertain points based on the definition of the cost function.

\begin{itemize}
\item {\bf{Unassigned version:}} \\
Here the goal is to find $k$ centers $C=\{c_1, \cdots, c_k\}$ that minimize
\begin{center}
$Ecost(c_1, c_2,\dots, c_k)=\sum_{R\in \Omega} prob(R)\max_{i=1,\dots,n}d(\hat{P}_i,C)$.
\end{center}

\item {\bf Unrestricted assigned version:}\\
 Here, all realizations of an uncertain point $P_i$ are assigned to a center denoted by $A(P_i)$. In fact, all realizations of an uncertain point $P_i$ in the assigned version are in the cluster of the same center. Therefore, the goal is to find $k$ centers $\{c_1, \cdots, c_k\}$ and an assignment  $A:\{P_1, \cdots,  P_n\} \rightarrow  \{c_1, \cdots, c_k\}$ that minimize
\begin{center}
$Ecost_A(c_1, c_2,\dots, c_k)=\sum_{R\in \Omega} prob(R)\max_{i=1,\dots,n}d(\hat{P}_i,A(P_i))$.
\end{center}
\item {\bf Restricted assigned version:}  
\\
Here for any set of uncertain points $\{P_1,\dots, P_n\}$ and $k$ centers $\{c_1,\dots,c_k\}$ an assignment 
$$A:\{P_1,\dots, P_n\}\rightarrow \{c_1,\dots, c_k\}$$
 is given. The goal is to find $\{c_1,\dots, c_k\}$ that minimizes $Ecost_A(c_1,\dots, c_k)$

In this paper, we consider three assignments: the expected distance assignment that was first introduced in \cite{wang}, the 1-center assignments and the expected point assignment for a Euclidean space where both of them are new in this paper as far as we know.

In the expected distance assignment, each uncertain point $P_i$ is assigned to $$ED(P_i)=arg.\min_{Q\in \{c_1\dots,c_k\}} \sum_{\hat{P}_i\in D_i} prob(\hat{P}_i)d(\hat{P}_i,Q).$$
 
 In a Euclidean space, let
 \begin{flalign*}
&\bar{P}_i=\sum_{\hat{P}_i\in D_i}prob(\hat{P}_i)\hat{P}_i.\\
\end{flalign*}
In the expected point assignment, each uncertain point $P_i$ is assigned to 
$$EP(P_i)=arg.\min_{Q\in \{c_1\dots,c_k\}}d(\bar{P_i},Q).$$

Finally, in the 1-center assignment, let $\tilde{P}_i$ be the 1-center of the single uncertain point $P_i$. An uncertain point $P_i$ is assigned to 
\begin{flalign*}
&OC(P_i)=arg.\min_{Q\in \{c_1\dots,c_k\}}d(\tilde{P}_i,Q).\\
 \end{flalign*}
\end{itemize}

%

\subsection*{Related works}
The deterministic $k$-center problem is a classical problem that has been extensively studied. It is well known that the $k$-center problem is NP-hard even in the plane \cite{meg} and approximation algorithms have been proposed (e.g., see \cite{arya,bad,har}). Efficient algorithms were also given for some special cases, e.g., the smallest enclosing circle and its weighed version and discrete version \cite{dyer,lee,meg28}, the Fermat-Weber problem \cite{chan}, $k$-center on trees \cite{chan9,fred,meg30}. Refer to \cite{drez} for other variations of facility location problems. The deterministic $k$-center in one-dimensional space is solvable in $O(n\log n)$ time \cite{meg31}.
One of the most elegant approximation algorithms for $k$-center clustering is the 2-factor approximation
algorithm by Gonzalez \cite{gon} which can be made to run in $O(n\log k)$ time \cite{fed}. One of the fastest methods for
$k$-center clustering in 2 and 3 dimensions is by Aggarwal and Procopiuc \cite{aga} which uses a dynamic programming
approach to $k$-center clustering and whose running time is upper bounded by $O(n\log k)+(\frac{k}{\epsilon})^{O(k^{1-\frac{1}{d}})}$.
Another elegant solution to the $k$-center clustering problem was given by Badoiu et.a \cite{bad}. This algorithm gives a $(1+\epsilon)$-approximation factor algorithm which runs in $2^{O((k\log k)/\epsilon^2)}dn$ in $\mathbf{R}^d$.
Another algorithm based on coresets runs in $O(k^{n})$ \cite{kumar} and it is claimed that the running time is much less than the worst case and thus it's possible to solve some problems when $k$ is small (say $k < 5$).

Several recent works have dealt with clustering problems on probabilistic data. One approach was to generalize well-known heuristic algorithms to the uncertain setting. For example a clustering algorithm called DBSCAN \cite{est} was also modified to handle probabilistic data by Kriegel and Pfeifle \cite{kri1,kri2} and Xu and Li \cite{xu}. Refer to \cite{caga} for a survey on data mining of uncertain data.

Cormode and McGregor \cite{cor} introduced the study of probabilistic clustering problems. They developed approximation algorithms for the probabilistic settings of $k$-means, $k$-median as well as $k$-center clustering. They described a pair of bicriteria approximation algorithms, for inputs of a particular form; one of which achieves a $(1+\epsilon)$-approximation with a large blow up in the number of centers, and the other which achieves a constant factor approximation with only $2k$ centers. 

Guha and Muhagala \cite{guh} improved upon the previous work. They achieved $O(1)$-approximations in finite metric space, while preserving the number of centers both for assigned and unassigned version of the $k$-center problem.
More precisely, the approximation factor of their algorithm for unrestricted assigned version is $15(1+2\epsilon)$ and the running time of their algorithm is polynomial in input size and $\dfrac{1}{\epsilon}$.

Munteanu and et.al. presented the first polynomial time (1 +$ \epsilon$)-approximation algorithm for the probabilistic smallest enclosing ball problem with extensions to the streaming setting\cite{mun} .

Wang and Zhang \cite{wang}, introduced the restricted assigned version under the expected distance assignment. They solved the one-dimensional $k$-center problem, in $O(zn \log zn + n \log k \log n)$ time. If dimension is one and the $z$ locations of each uncertain point are sorted, they reduced the problem to a linear programming problem and thus solved the problem in $O(zn)$ time by applying a linear time algorithm.

Haung and Li \cite{lin} gave a PTAS for unassigned version of the probabilistic $k$-center problem in $\mathbf{R}^d$, when both $k$ and $d$ are constants.

\section{Main results}

In this paper, we propose several approximation algorithms for restricted and unrestricted assigned version of uncertain $k$-center problem. In this section, we state the main results and in the next section, we give their proofs.

Our main approach is to replace each uncertain point $P$ with its expected point, $\bar{P}$, in the case of the Euclidean space or its 1-center, $\tilde{P}$,  in the case of the general metric space.
Next, we compute the $k$-center for the described certain points and prove that this solution gives an approximation solution for the uncertain points.
Note that there are efficient $(1+\epsilon)$-approximation algorithms for the certain $k$-center problem in the literature.
\subsection*{1-center in Euclidean space}
The first theorem gives a 2-approximation solution for the 1-center problem in the Euclidean space. 
\begin{theorem} \label{1}
Let $P_1,\dots, P_n$ be a set of uncertain points in the Euclidean space, and $$\bar{P}_1=\sum_{\hat{P}_1\in D_1} prob(\hat{P}_1)\hat{P}_1$$ be the expected point of $P_1$. Then $\bar{P}_1$ is a 2-approximation solution for the 1-center problem for $P_1,\dots, P_n$.
\end{theorem} 
Note that, we can compute $\bar{P}_1$ in $O(z)$ time which is independent of $n$.

\subsection*{Restricted assigned k-center problem in the Euclidean space}
For the restricted assigned $k$-center problem in the Euclidean space, we have the following theorem.

\begin{theorem}\label{2}
For a set of uncertain points $P_1,\dots, P_n$ in a Euclidean space, let $c_1,\dots, c_k$ be $(1+\epsilon)$-approximation solution for the $k$-center problem for $\bar{P}_1,\dots, \bar{P}_n$. Let $opt_{ED}$ and $op_{EP}$ be the minimum expected costs under the expected distance assignment and expected point assignment, respectively. Then,
\begin{flalign*}
&Ecost_{ED}(c_1,\dots,c_k)\le (5+\epsilon) opt_{ED}\\
\end{flalign*}
and
\begin{flalign*}
&Ecost_{EP}(c_1,\dots, c_k)\le (3+\epsilon) opt_{EP}.\\
\end{flalign*}
\end{theorem}
\subsection*{Unrestricted assigned k-center problem}
For unrestricted assigned $k$-center problem, we prove a stronger approximation algorithm for the Euclidean case and a slightly weaker one for a general metric space. 

For the unrestricted version, we present theorems that indicate the relation between the restricted assignment and unrestricted assignment.
 Note that in the unrestricted version, we have to compute the optimal $k$ centers and also the optimal assignment.

\begin{theorem}\label{3}
For a set of uncertain points $P_1,\dots, P_n$ in a metric space, the minimum expected cost under the expected distance assignment is a $3$-approximation for the minimum expected cost for the unrestricted assigned $k$-center problem.
\end{theorem}
So, any algorithm for the restricted assigned version under the expected point assignment gives a 3-approximation solution for the unrestricted assigned version. Since, the restricted assigned version under the expected distance assignment for $\mathbf{R}^1$ has exact solution \cite{wang}, so we have a 3-approximation solution for the unrestricted assigned version in $\mathbf{R}^1$.
For higher dimensions, we present the following theorems.
\begin{theorem}\label{4}
For a set of uncertain points $P_1,\dots, P_n$ in a Euclidean space, let $c_1,\dots, c_k$ be $(1+\epsilon)$-approximation solution for the $k$-center problem for $\bar{P}_1,\dots, \bar{P}_n$. Let $c_1^*,\dots, c_k^*$ and an assignment $A$ be the optimal solution for the unrestricted assigned $k$-center problem for $P_1,\dots, P_n$. Then,
\begin{flalign*}
&Ecost_{ED}(c_1,\dots,c_k)\le (5+\epsilon) Ecost_A(c_1^*,\dots, c_k^*).\\
\end{flalign*}
\end{theorem}
If, in the above theorem, instead of expected distance assignment we use the expected point assignment, then we get a better approximation factor.

\begin{theorem}\label{5}
For a set of uncertain points $P_1,\dots, P_n$ in a Euclidean space, let $c_1,\dots, c_k$ be $(1+\epsilon)$-approximation solution for the $k$-center problem for $\bar{P}_1,\dots, \bar{P}_n$. Let $c_1^*,\dots, c_k^*$ and an assignment $A$ be the optimal solution for the unrestricted assigned $k$-center problem for $P_1,\dots, P_n$. Then,
\begin{flalign*}
&Ecost_{EP}(c_1,\dots,c_k)\le (3+\epsilon) Ecost_A(c_1^*,\dots, c_k^*).\\
\end{flalign*}
\end{theorem}

In a general metric space, we do not have the expected point construction and instead we use the 1-center $\tilde{P}_i$ of the single uncertain point $P_i$.
\begin{theorem}\label{6}
For a set of uncertain points $P_1,\dots, P_n$ in a metric space, let $c_1,\dots, c_k$ be $(1+\epsilon)$-approximation solution for the $k$-center problem for $\tilde{P}_1,\dots, \tilde{P}_n$. Let $c_1^*,\dots, c_k^*$ and an assignment $A$ be the optimal solution for the unrestricted assigned $k$-center problem for $P_1,\dots, P_n$. Then,
\begin{flalign*}
&Ecost_{ED}(c_1,\dots,c_k)\le (7+2\epsilon) Ecost_A(c_1^*,\dots, c_k^*).\\
\end{flalign*}
\end{theorem}
If, in the above theorem, instead of expected distance assignment we use the 1-center assignment, then we get a better approximation factor.

\begin{theorem}\label{7}
For a set of uncertain points $P_1,\dots, P_n$ in a metric space, let $c_1,\dots, c_k$ be $(1+\epsilon)$-approximation solution for the $k$-center problem for $\tilde{P}_1,\dots, \tilde{P}_n$. Let $c_1^*,\dots, c_k^*$ and an assignment $A$ be the optimal solution for the unrestricted assigned $k$-center problem for $P_1,\dots, P_n$. Then
\begin{flalign*}
&Ecost_{OC}(c_1,\dots,c_k)\le (5+2\epsilon) Ecost_A(c_1^*,\dots, c_k^*).\\
\end{flalign*}
\end{theorem}

Note that the best constant approximation factor algorithm for the unrestricted assigned version, was $15(1+2\epsilon)$, with the polynomial running time in input size and $\dfrac{1}{\epsilon}$ \cite{guh}.

Our results are summarized  in Table \ref{maintable}. Note that the empty places for the running times are due to the fact that they depend on a $(1+\epsilon)$-approximation algorithm used for the $k$-center problem of certain points.
\begin{table*}[h]
\caption{Our results for various versions of uncertain $k$-center}
\begin{center}
\begin{tabular}{ccccl}
\toprule
Objective&Metric& Running time& Assignment& Approx-Factor \\ \hline
    \midrule

1-center&Euclidean&$O(z)$&-&2\\ \hline

$k$-center&Euclidean&$O(nz+n\log k)$&restricted assigned version expected distance&6\\ \hline

$k$-center&Euclidean&-&restricted assigned version expected distance&$5+\epsilon$\\ \hline

$k$-center&Euclidean&$O(nz+n\log k)$&restricted assigned version expected point&4\\ \hline

$k$-center&Euclidean&-&restricted assigned version expected point&$3+\epsilon$\\ \hline

$k$-center&Euclidean&$O(nz+n\log k)$&unrestricted assigned version&4\\ \hline

$k$-center&Euclidean&-&unrestricted assigned version&$3+\epsilon$\\ \hline

$k$-center&$\mathbf{R}^1$& $O(zn \log zn + n \log k \log n)$&unrestricted assigned version&3\\ \hline

$k$-center&any metric&-&unrestricted assigned version &$5+\epsilon$\\ \hline

\bottomrule
\end{tabular}
\end{center}
\label{maintable}
\end{table*}

\section{Proofs}

In this section, we provide the proofs of the theorems stated in the previous section. First, we present two lemmas that are crucial for the rest of this section.

\begin{lemma}\label{expp}
For an uncertain point $P$ in a Euclidean space and any point $Q$, we have
\begin{flalign*}
&d(\bar{P},Q)\le Ed(P,Q)=\sum_{\hat{P}\in D} prob(\hat{P})d(\hat{P},Q)\\
\end{flalign*}
where $\bar{P}=\sum_{\hat{P}\in D} prob(\hat{P})\hat{P}$ is the expected point of $P$.
\end{lemma}
\begin{proof}
Since, $d(\bar{P},Q)$ can be defined in terms of the Euclidean norm as $||\bar{P}-Q||$, using the triangle inequality
\begin{flalign*}
&||\bar{P}-Q||=||\sum_{\hat{P}\in D} prob(\hat{P})\hat{P}-Q||\\
&=||\sum_{\hat{P}\in D}prob(\hat{P})(\hat{P}-Q)||\le \sum_{\hat{P}\in D} prob(\hat{P})||\hat{P}-Q||\\
&=Ed(P,Q).\\
\end{flalign*}
\end{proof}

\begin{lemma}\label{exp}
For uncertain points $P_1,\dots, P_n$ , any $k$ centers $c_1,\dots, c_k$ and any assignment 
$A$,
we have
\begin{flalign*}
&Ecost_A(c_1,\dots, c_k)\geq \sum_{\hat{P}_1\in D_1}prob(\hat{P}_1)d(\hat{P}_1,A(P_1)).\\
\end{flalign*}
\end{lemma}
\begin{proof} Let $\Omega(\hat{P}_1)$ be those realizations that $P_1$ is realized as $\hat{P}_1$. Then, $\sum_{R\in \Omega(\hat{P}_1)} prob(R)=prob(\hat{P}_1)$. We have
\begin{flalign*}
&Ecost_A(c_1,\dots,c_k)=\sum_{R\in \Omega}prob(R)\max_{i=1\dots, n} d(\hat{P}_i,A(P_i))\\
&=\sum_{\hat{P}_1\in D_1}\sum_{R\in \Omega(\hat{P}_1)} prob(R)\max_{i=1\dots, n} d(\hat{P}_i,A(P_i))\\
&\geq \sum_{\hat{P}_1\in D_1}\sum_{R\in \Omega(\hat{P}_1)} prob(R) d(\hat{P}_1,A(P_1))\\
&=\sum_{\hat{P}_1\in D_1}prob(\hat{P}_1)d(\hat{P}_1,A(P_1)).\\
\end{flalign*}
\end{proof}

\subsection*{Proof of Theorem \ref{1}} Let $c^*$ be the optimal $1$-center of $P_1,\dots, P_n$, we need to show that
\begin{flalign*}
&Ecost(\bar{P}_1)\le 2 Ecost(c^*).\\
\end{flalign*}
By the definition of Ecost,
\begin{flalign*}
&Ecost(\bar{P}_1)=\sum_{R\in \Omega} prob(R)\max_{i=1,\dots, n} d(\bar{P}_1,\hat{P}_i).\\
\end{flalign*}
By triangle inequality,
\begin{flalign*}
&\le \sum_{R\in \Omega} prob(R)\max_{i=1,\dots, n} (d(\bar{P}_1,c^*)+d(c^*,\hat{P}_i))\\
&= d(\bar{P}_1,c^*)+\sum_R prob(R)\max_{i=1,\dots, n} d(c^*, \hat{P}_i).\\
\end{flalign*}
By Lemma \ref{expp} and definition of Ecost,
\begin{flalign*}
&\le (\sum_{\hat{P}_1\in D_1}prob(\hat{P}_1)d(\hat{P}_1,c^*))+Ecost(c^*).\\
\end{flalign*}
By Lemma \ref{exp},
\begin{flalign*}
&\le 2 Ecost(c^*).\\
\end{flalign*}

\subsection*{Proof of Theorem \ref{2}}
To prove Theorem \ref{2}, the following two lemmas are needed.
\begin{lemma}
\label{two}
For a set of uncertain points $P_1,\dots, P_n$ in a Euclidean space, let $c_1,\dots, c_k$ be any $k$ centers and  $A:\{P_1,\dots,P_n\}\rightarrow\{c_1,\dots,c_k\}$ be any assignment, we have
\begin{flalign*}
\sum_{R\in \Omega} prob(R)\max_{i=1,\dots,n} d(\hat{P}_i,\bar{P}_i)\leq 2Ecost_A(c_1,\dots,c_k),\\
\end{flalign*}
in particular for any $1\le i\le n$,
\begin{flalign*}
&\sum_{\hat{P}_i\in D_i} prob(\hat{P}_i)d(\hat{P}_i,\bar{P}_i)\leq 2Ecost_A(c_1,\dots,c_k).\\
\end{flalign*}
\end{lemma}

\begin{proof}
We have
\begin{flalign*}
&\sum_{R\in \Omega} prob(R)\max_{i=1,\dots,n} d(\hat{P}_i,\bar{P}_i)\\
&\leq \sum_{R\in \Omega} prob(R)\max_{i=1,\dots,n} (d(\hat{P}_i,A({P_i})+d(A(P_i),\bar{P}_i))\\
&\leq \sum_{R\in \Omega} prob(R)\max_{i=1,\dots,n} d(\hat{P}_i,A(P_i))+d(A(P_1),\bar{P}_1)\\
\end{flalign*}
where, we assume $d(A(P_1),\bar{P}_1)=\max_{i=1,\dots,n}d(A(P_i),\bar{P}_i)$, now the above term is

\begin{flalign*}
&=Ecost_A(c_1,\dots,c_k)+d(A(P_1),\bar{P}_1).\\
\end{flalign*}
It is enough to show $d(A(P_1),\bar{P}_1)\leq Ecost_A(c_1,\dots,c_k)$. But, according to Lemma \ref{exp}, we have
\begin{flalign*}
d(A(P_1),\bar{P}_1)\leq \sum_{\hat{P}_1\in D_1} prob(\hat{P}_1)d(A(P_1),\hat{P}_1).\\
\end{flalign*}
and from Lemma \ref{exp}, it follows that $d(A(P_1),\bar{P}_1)\leq Ecost_A(c_1,\dots, c_k)$.
\end{proof}
\begin{lemma}\label{three}
Let $P_1,\dots, P_n$ be a set of uncertain points for any $k$ centers $c_1,\dots, c_k$ and assignment $A$ one has
\begin{flalign*}
&cost(c_1,\dots, c_k)\le Ecost_A(c_1,\dots, c_k).\\
\end{flalign*}
where cost is for the certain points $\bar{P}_1,\dots, \bar{P}_n$.
\end{lemma}
\begin{proof}
One has
\begin{flalign*}
&cost(c_1,\dots, c_k)=d(c_i,\bar{P}_j).\\
\end{flalign*}
Since, $c_i$ is the closest center to $\bar{P}_j$ and by Lemma \ref{expp},
\begin{flalign*}
&\le d(A(P_j),\bar{P}_j)\le \sum_{\hat{P}_j\in D_j} prob(\hat{P}_j)d(\hat{P}_j, A(P_j))\\
\end{flalign*}
and by Lemma \ref{exp},
\begin{flalign*}
&\le Ecost_A(c_1,\dots, c_k).\\
\end{flalign*}
\end{proof}

Now, we present the proof of Theorem \ref{2} for the expected distance assignment.  Let $c^*_1,\dots, c^*_k$ be the optimal solution for restricted assigned version of $k$-center problem with the expected distance assignment. We need to show
\begin{flalign*}
&Ecost_{ED}(c_1,\dots,c_k)\le (5+\epsilon) Ecost_{ED}(c^*_1,\dots,c^*_k).\\
\end{flalign*}
By definition,
\begin{flalign*}
&Ecost_{ED}(c_1,\dots, c_k)=\sum_R prob(R)\max_{i=1,\dots, n} d(\hat{P}_i, ED(P_i)).\\
\end{flalign*}
By triangle inequality,
\begin{flalign*}
&\le \sum_R prob(R)\max_{i=1,\dots, n} (d(\hat{P}_i, \bar{P}_i)+d(\bar{P}_i,ED(P_i))).\\
\end{flalign*}
If we let $d(\bar{P}_1,ED(P_1))=\max_{i=1\dots, n} d(\bar{P}_i,ED(P_i))$ and use Lemma \ref{two}, 
\begin{flalign*}
&\le 2Ecost_{ED}(c^*_1,\dots,c^*_k)+d(\bar{P}_1,ED(P_1))\\
\end{flalign*}
So, we need to show that
\begin{flalign*}
&d(\bar{P}_1,ED(P_1))\le (3+\epsilon)Ecost_{ED}(c^*_1,\dots,c^*_k).\\
\end{flalign*}
Let $c_i$ be the closest point among $\{c_1,\dots, c_k\}$ to $\bar{P}_1$. Then, by Lemma \ref{expp},
\begin{flalign*}
&d(\bar{P}_1,ED(P_1))\le \sum_{\hat{P}_1\in D_1} prob(\hat{P}_1)d(\hat{P}_1,ED(P_1)).\\
\end{flalign*}
Since, $ED(P_1)$ has the closest expected distance to $P_1$ among $c_1,\dots, c_k$, 
\begin{flalign*}
&\le \sum_{\hat{P}_1\in D_1} prob(\hat{P}_1)d(\hat{P}_1,c_i).\\
\end{flalign*}
By triangle inequality,
\begin{flalign*}
&\le (\sum_{\hat{P}_1\in D_1} prob(\hat{P}_1)d(\hat{P}_1,\bar{P}_1))+d(\bar{P}_1,c_i).\\
\end{flalign*}
By Lemma \ref{two},
\begin{flalign*}
&\sum_{\hat{P}_1} prob(\hat{P}_1)d(\hat{P}_1,\bar{P}_1)\le 2Ecost_{ED}(c^*_1,\dots, c^*_k).\\
\end{flalign*}
So, it remains to show
\begin{flalign*}
&d(\bar{P}_1,c_i)\le (1+\epsilon)Ecost_{ED}(c^*_1,\dots, c^*_k).\\
\end{flalign*}
Since, $c_i$ is the closest center to $\bar{P}_1$ we have
\begin{flalign*}
&d(\bar{P}_1,c_i)\le cost(c_1,\dots,c_k)\\
\end{flalign*}
where cost is for the certain points $\bar{P}_1,\dots, \bar{P}_n$. Since, $c_1,\dots, c_k$ is a $(1+\epsilon)$-approximation solution for the $k$-center problem,
\begin{flalign*}
&\le (1+\epsilon)cost(c^*_1,\dots, c^*_k)\\
\end{flalign*}
and by Lemma \ref{three},
\begin{flalign*}
&\le (1+\epsilon) Ecost_{ED}(c^*_1,\dots,c^*_k).
\end{flalign*}
So, Theorem \ref{2} for the expected distance assignment is proved.

Now, we give the proof of Theorem \ref{2} for the the expected point assignment.
Let $c^*_1,\dots, c^*_k$ be the optimal solution for the restricted assigned $k$-center problem for the expected point assignment. We need to show
\begin{flalign*}
&Ecost_{EP}(c_1,\dots, c_k)\le (3+\epsilon) Ecost_{EP}(c^*_1,\dots, c^*_k).\\
\end{flalign*}
By definition,
\begin{flalign*}
&Ecost_{EP}(c_1,\dots, c_k)=\sum_R prob(R)\max_{i=1,\dots, n} d(\hat{P}_i, EP(P_i)).\\
\end{flalign*}
By triangle inequality,
\begin{flalign*}
&\le \sum_R prob(R)\max_{i=1,\dots, n} (d(\hat{P}_i, \bar{P}_i)+d(\bar{P}_i,EP(P_i)).\\
\end{flalign*}
If we let $d(\bar{P}_1,ED(P_1))=\max_{i=1\dots, n}d(\bar{P}_i,ED(P_i))$ and use Lemma \ref{two},
\begin{flalign*}
&\le 2Ecost_{ED}(c^*_1,\dots,c^*_k)+d(\bar{P}_1,EP(P_1)).\\
\end{flalign*}
So, we need to show that
\begin{flalign*}
&d(\bar{P}_1,EP(P_1))\le (1+\epsilon)Ecost_{ED}(c^*_1,\dots,c^*_k).\\
\end{flalign*}
By definition of expected point assignment,
\begin{flalign*} 
&d(\bar{P}_1,EP(P_1))=cost(c_1,\dots, c_k)\\
\end{flalign*}
 and since, $c_1,\dots, c_k$ is a $(1+\epsilon)$-approximation solution,
 \begin{flalign*}
&cost(c_1,\dots,c_k)\le (1+\epsilon)cost(c^*_1,\dots, c^*_k)\\
\end{flalign*}
and by Lemma \ref{three},
\begin{flalign*}
&\le (1+\epsilon)Ecost_{EP}(c^*_1,\dots,c^*_k).\\
\end{flalign*}
So, Theorem \ref{2} is completely proved.

\begin{remark} There is a greedy 2-approximation algorithm for deterministic $k$-center problem of certain points $\bar{P}_1,\dots, \bar{P}_n$ in a metric space given in \cite{gon}. It is as follows. First, choose any point, say $\bar{P}_1$ and then choose the farthest point from $\bar{P}_1$, say $\bar{P}_2$ and then, the farthest point from the set $\{\bar{P}_1,\bar{P}_2\}$, say $\bar{P}_3$ and continue until finding the farthest point from the set $\{\bar{P}_1,\dots,\bar{P}_{k-1}\}$, say $\bar{P}_k$. Then, the points $\bar{P}_1,\dots, \bar{P}_k$ is a 2-approximation solution for the deterministic $k$-center problem. If we use this method,
in the first phase of the algorithm, we compute the expected point of each probabilistic point which takes $O(nz)$. Next, we compute $\bar{P}_1,\dots,\bar{P}_k$, The running time of this phase is $O(n\log k)$ \cite{fed}.
So, the overall running time of algorithm is $O(nz+n\log k)$ and we get respectively a 6 and 4 approximation for the optimal expected cost of the $k$-center problem for the expected distance and expected point assignments.

\end{remark}
\subsection*{Proof of Theorem \ref{3}}
 Let $c_1,\dots, c_k$ be the optimal solution for the restricted assigned $k$-center problem with expected distance assignment. Let $c^*_1,\dots, c^*_k$ and assignment $A$ be the optimal solution for the unrestricted assigned $k$-center problem. Then,
 \begin{flalign*}
&Ecost_{ED}(c_1,\dots, c_k)\le Ecost_{ED}(c^*_1,\dots,c^*_k)\\
&=\sum_R prob(R)\max_{i=1,\dots, n} d(\hat{P}_i, ED(P_i))\\
&\le \sum_R prob(R)\max_{i=1,\dots, n}(d(\hat{P}_i,A(P_i))+d(A(P_i),ED(P_i)))\\
&\le Ecost_A(c^*_1,\dots, c^*_k)+d(A(P_1),ED(P_1))\\
\end{flalign*}
where $d(A(P_1),ED(P_1))=\max_{i=1\dots, n} d(A(P_i),ED(P_i))$.
By triangle inequality,
\begin{flalign*}
&d(A(P_1),ED(P_1))\\
&\le \sum_{\hat{P}_1\in D_1} prob(\hat{P}_1)(d(A(P_1),\hat{P}_1)+d(\hat{P}_1, ED(P_1))\\
\end{flalign*}
By Lemma \ref{exp} and the fact that $ED(P_1)$ has the smallest expected distance from $P_1$ among $c^*_1,\dots, c^*_k$, we get
\begin{flalign*}
&\le Ecost_A(c^*_1,\dots, c^*_k)+\sum_{\hat{P}_1\in D_1} prob(\hat{P}_1)d(\hat{P}_1,A(P_1))\\
&\le 2 Ecost_A(c^*_1,\dots,c^*_k).\\
\end{flalign*}
So, Theorem \ref{3} is proved.
\subsection*{Proof of Theorem \ref{4}}

By definition,
\begin{flalign*}
&Ecost_{ED}(c_1,\dots, c_k)=\sum_R prob(R)\max_{i=1\dots, n} d(\hat{P}_i, ED(P_i)).\\
\end{flalign*}
By triangle inequality,
\begin{flalign*}
&\le \sum_R prob(R)\max_{i=1\dots, n}(d(\hat{P}_i, A(P_i))+d(A(P_i),ED(P_i)))\\
&\le Ecost_A(a_1,\dots,a_k)+d(A(P_1),ED(P_1))\\
\end{flalign*}
where $d(A(P_1),ED(P_1))=\max_{i=1\dots, n} d(A(P_i),ED(P_i))$.
Now by triangle inequality and Lemma \ref{expp},
\begin{flalign*}
&d(A(P_1),ED(P_1))\le d(A(P_1),\bar{P}_1)+d(\bar{P}_1,ED(P_1))\\
&\le \sum_{\hat{P}_1}prob(\hat{P}_1)d(A(P_1),\hat{P}_1)+\sum_{\hat{P}_1}prob(\hat{P}_1)d(\hat{P}_1,ED(P_1)).\\
\end{flalign*}
Let $c_1$ be a center among $c_1,\dots, c_k$ that is closest to $\bar{P}_1$. By Lemma \ref{exp} and the fact that $ED(P_1)$ has the closest expected distance to $P_1$ among the centers we get
\begin{flalign*}
&\le Ecost_A(c^*_1,\dots ,c^*_k)+\sum_{\hat{P}_1} prob(\hat{P}_1)d(\hat{P}_1,c_1).\\
\end{flalign*}
If instead of $d(\hat{P}_1,c_1)$, we put $d(\hat{P}_1,A(P_1))+d(A(P_1),c_1)$ and use Lemma \ref{exp}, we get
\begin{flalign*}
&\le 2 Ecost_A(c^*_1,\dots, c^*_k)+d(A(P_1),c_1).\\
\end{flalign*}
Now,
\begin{flalign*}
&d(A(P_1),c_1)\le d(A(P_1),\bar{P}_1)+d(\bar{P}_1,c_1)\\
&\le \sum_{\hat{P}_1}prob(\hat{P}_1)d(\hat{P}_1,A(P_1))+cost(c_1,\dots, c_k).\\
\end{flalign*}
By Lemma \ref{exp} and the fact that $c_1,\dots,c_k$ is a $(1+\epsilon)$-approximation solution for the $k$-center problem,
\begin{flalign*}
&\le Ecost_A(c^*_1,\dots, c^*_k)+(1+\epsilon)cost(c^*_1,\dots,c^*_k).\\
\end{flalign*}
Finally, by Lemma \ref{three},
\begin{flalign*}
&\le (2+\epsilon)Ecost_A(c^*_1,\dots,c^*_k).\\
\end{flalign*}
This proves Theorem \ref{4}.
\subsection*{Proof of Theorem \ref{5}}
By definition,
\begin{flalign*}
&Ecost_{EP}(c_1,\dots, c_k)\\
&=\sum_{R\in \Omega} prob(R)\max_{i=1,\dots, n} d(\hat{P}_i, EP(P_i))\\
&\le \sum_{R\in \Omega} prob(R)\max_{i=1,\dots, n} (d(\hat{P}_i,\bar{P}_i)+d(\bar{P}_i,EP(P_i))).\\
\end{flalign*}
If we let $d(\bar{P}_1,EP(P_1))=\max_{i=1\dots, n}d(\bar{P}_i,EP(P_i))$ and use Lemma \ref{two} we get
\begin{flalign*}
&\le 2Ecost_A(c^*_1,\dots,c^*_k)+d(\bar{P}_1,EP(P_1)).\\
\end{flalign*}
Now, by Lemma \ref{expp},
\begin{flalign*}
\\d(\bar{P}_1, EP(P_1))\le \sum_{\hat{P}_1}prob(\hat{P}_1)d(\hat{P}_1,EP(P_1)).\\
\end{flalign*}
Since, 
\begin{flalign*}
&d(\hat{P}_1,EP(P_1))=cost(c_1,\dots ,c_k)\\
&\le (1+\epsilon)cost(c^*_1,\dots,c^*_k),\\
\end{flalign*}
also by Lemma \ref{three},
\begin{flalign*}
&\le (1+\epsilon)Ecost_A(c^*_1,\dots,c^*_k),\\
\end{flalign*}
this proves Theorem \ref{5}.
\subsection*{Proofs of Theorem \ref{6} and Theorem \ref{7}}

To prove theorems \ref{6} and \ref{7}, we need two lemmas that are analogue of Lemmas \ref{two} and \ref{three} for a metric space.
\begin{lemma}\label{four}
Let $P_1,\dots, P_n$ be a set of uncertain points in a metric space. Let $\tilde{P}_i$ be the 1-center for the single uncertain point $P_i$. For any set of centers $c_1,\dots, c_k$ and any assignment $A:\{P_1,\dots, P_n\}\rightarrow \{c_1,\dots, c_k\}$ we have
\begin{flalign*}
&\sum_R prob(R)\max_{i=1,\dots, n}d(\hat{P}_i,\tilde{P}_i)\le 3 Ecost_A(c_1,\dots,c_k).\\
\end{flalign*}
\end{lemma}
\begin{proof}
Let $d(A(P_1),\tilde{P}_1)=\max_{i=1\dots, n} d(A(P_i),\tilde{P}_i)$. If we use $d(\hat{P}_i,\tilde{P}_i)\le d(\hat{P}_i,A(P_i))+d(A(P_i),\tilde{P}_i)$, we get that the left hand side is
\begin{flalign*}
&\le Ecost_A(c_1,\dots,c_k)+d(A(P_1),\tilde{P}_1).\\
\end{flalign*}
By triangle inequality,
\begin{flalign*}
&d(A(P_1),\tilde{P}_1)\\
&\le \sum prob(\hat{P}_1)d(A(P_1),\hat{P}_1)+\sum prob(\hat{P}_1)d(\hat{P}_1,\tilde{P}_1).\\
\end{flalign*}
Since, $\tilde{P}_1$ is 1-center we get
\begin{flalign*}
&\le 2\sum prob(\hat{P}_1)d(A(P_1),\hat{P}_1),\\
\end{flalign*}
and by Lemma \ref{exp},
\begin{flalign*}
&\le 2 Ecost_A(c_1,\dots, c_k).\\
\end{flalign*}
This proves the lemma.
\end{proof}

\begin{lemma}\label{five}
Let $P_1,\dots, P_n$ be a set of uncertain points in a metric space. For any $k$ centers $c_1,\dots, c_k$ and assignment $A$ one has
\begin{flalign*}
cost(c_1,\dots, c_k)\le 2Ecost_A(c_1,\dots, c_k).\\
\end{flalign*}
where cost is for the certain points $\tilde{P}_1,\dots, \tilde{P}_n$, where $\tilde{P}_i$ is the 1-center of the uncertain point $P_i$.
\end{lemma}
\begin{proof}
Let 
$$cost(c_1,\dots,c_k)=d(c_i,\tilde{P}_j)$$
Then, since $c_i$ is the closest center to $\tilde{P}_j$,
\begin{flalign*}
&\le d(A(P_j),\tilde{P}_j)\\
\end{flalign*}
by triangle inequality,
\begin{flalign*}
&\le \sum prob(\hat{P}_j)d(A(P_j),\hat{P}_j)+\sum prob(\hat{P}_j)d(\hat{P}_j,\tilde{P}_j).\\
\end{flalign*}
since, $\tilde{P}_j$ is 1-center of $P_j$,
\begin{flalign*}
&\le 2\sum prob(\hat{P}_j)d(A(P_j),\hat{P}_j)\\
&\le 2Ecost_A(c_1,\dots,c_k)\\
\end{flalign*}
So, the lemma is proved.
\end{proof}
We now prove Theorem \ref{6}.
By definition,
\begin{flalign*}
&Ecost_{ED}(c_1,\dots, c_k)=\sum_{R\in \Omega}prob(R)\max_{i=1\dots, n} d(\hat{P}_i,ED(P_i))\\
&\leq \sum_{R\in \Omega}prob(R)\max_{i=1\dots, n} d(\hat{P}_i,\tilde{P}_i)+d(\tilde{P}_1,ED(P_1)).\\
\end{flalign*}
Where $d(\tilde{P}_1,ED(P_1))=\max_{i=1,\dots k}d(\tilde{P}_i,ED(P_i))$. Since, by Lemma \ref{four}, the first term is at most $3Ecost_A(c^*_1,\dots, c^*_k)$, it is enough to show $$d(\tilde{P}_1,ED(P_1))\leq (4+2\epsilon)Ecost_A(c^*_1,\dots, c^*_k).$$ Now by triangle inequality and the fact that $\tilde{P}_1$ is 1-center of $P_1$ we get
\begin{flalign*}
&d(\tilde{P}_1,ED(P_1))\\
&\leq \sum_{\hat{P}_1\in D_1} prob(\hat{P}_1)(d(\tilde{P}_1,\hat{P}_1)+d(\hat{P}_1,ED(P_1)))\\
&\leq \sum_{\hat{P}_1\in D_1} prob(\hat{P}_1)\left(d(A(P_1),\hat{P}_1)+d(\hat{P}_1,ED(P_1))\right)\\
&\leq Ecost_A(c^*_1,\dots ,c^*_k)+\sum_{\hat{P}_1\in D_1}  prob(\hat{P}_1)d(\hat{P}_1,ED(P_1))\\
&\leq Ecost_A(c^*_1,\dots ,c^*_k)+\sum_{\hat{P}_1\in D_1}  prob(\hat{P}_1)d(\hat{P}_1,c_j)\\
\end{flalign*}
where $c_j$ is the closest among $c_i$'s to $\tilde{P}_1$. Now,
\begin{flalign*}
&\sum_{\hat{P}_1\in D_1}  prob(\hat{P}_1)d(\hat{P}_1,c_j)\\
&\leq \sum_{\hat{P}_1\in D_1}  prob(\hat{P}_1)d(\hat{P}_1,\tilde{P_1}))+d(\tilde{P_1},c_j)\\
&\leq \sum_{\hat{P}_1\in D_1}  prob(\hat{P}_1)d(\hat{P}_1,A(P_1))+d(\tilde{P}_1,c_j)\\
&\leq Eِِِِِِcost_A(c^*_1,\dots ,c^*_k)+d(\tilde{P_1},c_j).\\
\end{flalign*}
Now, $d(\tilde{P_1},c_j)\le cost(c_1,\dots, c_k)$. Since, these centers are a $(1+\epsilon)$-approximation solution for the $k$-center problem,
\begin{flalign*}
&cost(c_1,\dots, c_k)\\
&\le (1+\epsilon) cost(c^*_1,\dots, c^*_k)\\
\end{flalign*}
by lemma \ref{five},
\begin{flalign*}
\le (2+2\epsilon)Ecost_A(c^*_1,\dots,c^*_k)\\
\end{flalign*}
and this finishes the proof of Theorem \ref{6}.

Finally, we prove Theorem \ref{7}.
By definition,
\begin{flalign*}
&Ecost_{OC}(c_1,\dots, c_k)=\sum_R prob(R)\max_{i=1\dots, n} d(P_i, OC(P_i)).\\
\end{flalign*}
By triangle inequality,
\begin{flalign*}
&\le \sum_R prob(R)\max_{i=1\dots, n} (d(P_i, \tilde{P}_i)+d(\tilde{P}_i,OC(P_i))).\\
\end{flalign*}
By Lemma \ref{four},
\begin{flalign*}
&\le 3Ecost_A(c^*_1,\dots, c^*_k)+d(\tilde{P}_1,OC(P_1))\\
\end{flalign*}
where $d(\tilde{P}_1,OC(P_1))=\max_{i=1\dots, n} d(\tilde{P}_i,OC(P_i))$.
Now,
\begin{flalign*}
&d(\tilde{P}_1,OC(P_1))=cost(c_1,\dots,c_k)\le (1+\epsilon)cost(c^*_1,\dots, c^*_k)\\
\end{flalign*}
and by lemma \ref{five},
\begin{flalign*}
 \le (2+2\epsilon)Ecost_A(c^*_1,\dots,c^*_k)
 \end{flalign*}
and this finishes the proof.

\section{conclusion}
In this paper the $k$-center problem for uncertain data points have been studied. We have proposed new assignment schemes and obtained improved constant approximation factor algorithms for them. Note that, the new assignments introduced in this paper allowed us to improve the approximation factor for the unrestricted assigned version.

The restricted version with expected distance assignment for $\mathbf{R}^1$ was studied in \cite{wang}. Here we gave approximation algorithms for $\mathbf{R}^d$ and also for any metric space.

The case of unrestricted assigned version which was studied in \cite{guh}, has been improved. The constant of approximation has been reduced to $5+\epsilon$  from $15+\epsilon$. We have also separately studied the case for the metric space and the Euclidean space.
In a future work, we intend to use our approach to study the $k$-median and the $k$-mean problems.

 Also, we intend to give a PTAS for the assigned versions of the uncertain $k$-center problem. 

\section*{Acknowledgment}
 The authors would like to thank Dr. Mohammad Ali Abam who introduced the problem to them and gave valuable comments and helpful suggestions.